\documentclass{article}%
\usepackage{amssymb}
\usepackage{amsfonts}
\usepackage{amsmath}
\usepackage{revsymb}
\usepackage{hyperref}%
\usepackage{graphicx}
\providecommand{\U}[1]{\protect\rule{.1in}{.1in}}
\newtheorem{theorem}{Theorem}

\newtheorem{definition}[theorem]{Definition}
\newtheorem{example}[theorem]{Example}

\newtheorem{lemma}[theorem]{Lemma}

\newenvironment{proof}[1][Proof]{\noindent\textbf{#1.} }{\ \rule{0.5em}{0.5em}}
\begin{document}

\title{Estimates of non-optimality of quantum measurements and a simple iterative
method for computing optimal measurements}
\author{Jon Tyson\thanks{jonetyson@X.Y.Z$,$ where X=post Y=harvard, and Z=edu}\\Harvard University}
\date{Feb 14, 2008}
\maketitle

\begin{abstract}
\noindent We construct crude estimates for non-optimality of quantum
measurements in terms of their violation of Holevo's simplified minimum-error
optimality conditions. As an application, we show that a modification of
Barnett and Croke's proof of the optimality conditions yields a convergent
iterative scheme for computing optimal measurements.

\end{abstract}

\newpage

\section{Introduction}

The \textit{minimum-error quantum detection problem}\ arose in the 1960's in
the design of optical detectors $\cite{Helstrom Quantum Detection and
Estimation Theory}$ and has been of recent importance in the subjects of
quantum information $\cite{pure state HSW theorem, mixed state HSW theorem,
Holevo mixed state HSW theorem, Koenig Renner Schaffner operational meaning of
min and max entropy}$ and quantum computation $\cite{Ip Shor's algorithm is
optimal, Bacon, Childs from optimal to efficient algo, optimal alg for hidden
shift, moore and russels distinguishing, Bacon new hidden subgroup}$:

\begin{quote}
If an unknown state $\rho_{k}$ is randomly chosen from a known ensemble of
quantum states, what is the chance that the value of $k$ will be discovered by
an optimal measurement?
\end{quote}

Barnett and Croke \cite{Barnett Croke on the conditions for discrimination
between quantum states with minimum error} have recently provided a simple
operator-theoretic proof of the necessity of the standard Yuen-Kennedy-Lax \&
Holevo (YKLH) optimality conditions \cite{Yuen Ken Lax Optimum testing of
multiple, Holevo optimal measurement conditions 1} for the minimum-error
quantum detection problem. Their proof may be shortened, since Holevo
\cite{Holevo remarks on optimal measurements} had previously shown that an
intermediate step of their proof (positivity of the operators $\hat{G}_{j}$
defined by equation $\left(  10\right)  $ of \cite{Barnett Croke on the
conditions for discrimination between quantum states with minimum error})
provides a simplified necessary and sufficient condition for minimum-error
quantum detection.

\subsection{Results}

This note gives a more robust version of Holevo's simplified optimality
condition (condition II of Theorem
\ref{theorem simplified optimality condition}, below), by estimating
non-optimality in terms of quantitative violation of this condition. These
bounds are used to show that the perturbative method of Barnett and Croke may
be converted into a convergent iterative algorithm for computing optimal
measurements, adding to the list $\cite{Helstrom Bayes cost reduction, Jezek
Rehacek and Fiurasek Finding optimal strategies for minimum error quantum
state discrimination, Hradil et al Maximum Likelihood methods in quantum
mechanics, eldar short and sweet optimal measurements}$ of algorithms for this
purpose. This iteration converges even for countably-infinite ensembles in an
infinite-dimensional Hilbert space.

\section{Conditions for minimum-error quantum discrimination}

A precise description of the minimum-error quantum measurement problem is
given by:

\begin{definition}
\label{main BC definition}Let $\mathcal{E}=\left\{  \rho_{k}\right\}  _{k\in
K}$ be an ensemble of mixed quantum states $\rho_{k}$, which are represented
as positive semidefinite operators on a Hilbert space $\mathcal{H}$ normalized
by a-priori probability: $\operatorname*{Tr}\rho_{k}=p_{k}$ with $%
{\displaystyle\sum}
p_{k}=1$. The \textbf{support} $\operatorname{supp}\left(  \mathcal{E}\right)
$ is the closure of the span of the ranges of the $\rho_{k}$. A
\textbf{positive operator-valued measurement} \textbf{(POVM) }is a collection
of positive semidefinite operators $\left\{  M_{k}\right\}  $ satisfying $%
{\displaystyle\sum}
M_{k}=%
\openone
$. The corresponding \textbf{Lagrange operator} is given by%
\begin{equation}
L=%
{\displaystyle\sum}
M_{k}\rho_{k}\text{.\label{def of L}}%
\end{equation}
The \textbf{minimum-error quantum discrimination problem} \cite{Helstrom
Quantum Detection and Estimation Theory} consists of finding a POVM maximizing
the \textbf{success probability}%
\begin{equation}
P_{\text{succ}}\left(  \left\{  M_{k}\right\}  \right)  =\operatorname*{Tr}%
{\displaystyle\sum_{k}}
M_{k}\rho_{k}=\operatorname*{Tr}L\label{sucess prob}%
\end{equation}
of correctly distinguishing an element blindly drawn from the ensemble
$\mathcal{E}$. (We will often abuse notation by writing $P_{\text{succ}%
}\left(  M_{k}\right)  $ instead of $P_{\text{succ}}\left(  \left\{
M_{k}\right\}  \right)  $.)
\end{definition}

Holevo's simplified optimality conditions are given by property II
of\footnote{Another interesting optimality condition is given by Theorem 3 on
page 39 of \cite{Belavkin Book}.} 

\begin{theorem}
[Holevo \cite{Holevo remarks on optimal measurements}, Yuen-Kennedy-Lax
\cite{Yuen Ken Lax Optimum testing of multiple}, ]Let $\left\{  M_{k}\right\}
_{k=1,...,m}$ be a POVM for distinguishing the ensemble $\mathcal{E}$. Then
the following are equivalent:\label{theorem simplified optimality condition}

\begin{enumerate}
\item[I.] \label{optimality stated}$\left\{  M_{k}\right\}  $ maximizes
$P_{\text{succ}}$.

\item[II.] \label{condition Yuen Holevo redundant}$\left(  L+L^{^{\dag}%
}\right)  /2\geq\rho_{k}$ for all $k$.\footnote{Earlier formulations of
condition II \cite{Yuen Ken Lax Optimum testing of multiple, Holevo optimal
measurement conditions 1} were that $L=L^{^{\dag}}$ and $L\geq\rho_{k}$ for
all $k,$ equivalently stated as $L=L^{^{\dag}}$ and $\operatorname{Re}\left(
L\right)  \geq\rho_{k}$. (The self-adjointess condition is redundant in the
latter form.)}

\item[III.] \label{condition bounded above by G}There exists a self-adjoint
operator $G$ satisfying $G\geq\rho_{k}$ and $\left(  G-\rho_{k}\right)
M_{k}=0$ for all $k$.
\end{enumerate}

\noindent Furthermore, under these equivalent conditions $L=L^{^{\dag}}=G$,
and $L$ is the unique self-adjoint operator of minimal trace satisfying
$L\geq\rho_{k}$ for all $k$.
\end{theorem}

The above optimality conditions were first proved in the infinite-dimensional
case by Holevo, since earlier proofs worked only in finite dimensions. The
inequalities in properties II-III use the standard order on self-adjoint
matrices:\ $A\geq B$ iff $A-B$ is positive semidefinite. The LHS of condition
II is commonly referred to as the \textbf{real part}:
\begin{equation}
\operatorname{Re}\left(  L\right)  :=\left(  L+L^{\dag}\right)
/2.\label{def real part}%
\end{equation}

\section{Mathematical background}

\begin{definition}
\label{def of positive part}Let $A$ be a self-adjoint operator on a Hilbert
space $\mathcal{H}$ with spectral decomposition $A=%
{\displaystyle\sum}
\lambda_{k}\Pi_{k} $. The \textbf{positive part} of $A$ is given by
\begin{equation}
\left[  A\right]  _{+}=%
{\displaystyle\sum_{\lambda_{k}>0}}
\lambda_{k}\Pi_{k}\text{.\label{def of positive pt}}%
\end{equation}
The \textbf{positive projection} is given by%
\begin{equation}
\chi_{+}\left(  A\right)  =%
{\displaystyle\sum_{\lambda_{k}>0}}
\Pi_{k}\text{.\label{def of positive projection}}%
\end{equation}
The \textbf{trace norm }of an operator $B:\mathcal{H}\rightarrow\mathcal{H}$
is given by $\left\Vert B\right\Vert _{1}=\operatorname*{Tr}\sqrt{B^{\dag}B}$.
The \textbf{operator norm} is given by $\left\Vert B\right\Vert _{\infty}%
=\sup_{\left\Vert \psi\right\Vert =1}\left\Vert B\psi\right\Vert $.
\end{definition}

We collect some simple mathematical facts. We will frequently use the
inequalities
\begin{align}
\left\vert \operatorname*{Tr}A\right\vert  & \leq\left\Vert A\right\Vert
_{1}\label{inequality trace less than i1 norm}\\
\left\Vert BC\right\Vert _{1}  & \leq\left\Vert B\right\Vert _{1}\left\Vert
C\right\Vert _{\infty},\label{inequality l1 linfinity holder}%
\end{align}
which may be found in \cite{Reed and Simon I}. For positive semidefinite
operators $P_{1},P_{2}\geq0$ such that $P_{1}P_{2}$ is trace class, one has%
\begin{equation}
\operatorname*{Tr}P_{1}P_{2}\geq0,\label{fact yuen trace of P1P2}%
\end{equation}
with equality iff $P_{1}P_{2}=0$ \cite{Yuen Ken Lax Optimum testing of
multiple} and%
\begin{equation}
A_{1}\geq A_{2}\Rightarrow C^{\dag}A_{1}C\geq C^{\dag}A_{2}%
C\label{fact conjugate a positve operator}%
\end{equation}
for all operators $C$ and self-adjoint $A_{1},A_{2}$.

\section{\label{section robust}Estimates of near- and non-optimality}

Our next goal is to strengthen condition II\ of Theorem
\ref{theorem simplified optimality condition} by giving quantitative bounds in
the case that condition II fails to hold. As a first step, note that in the
finite-dimensional case if%
\begin{equation}
\operatorname{Re}\left(  L\right)  \geq\rho_{k}-\alpha
\end{equation}
for some scalar $\alpha>0$, then by inequality $\left(
\ref{fact yuen trace of P1P2}\right)  $%
\begin{align}
P_{\text{succ}}\left(  M_{k}\right)    & =\operatorname*{Tr}\operatorname{Re}%
\left(  L\right)  =\operatorname*{Tr}%
{\displaystyle\sum_{k}}
\operatorname{Re}\left(  L\right)  M_{k}^{\text{opt}}\\
& \geq\operatorname*{Tr}%
{\displaystyle\sum_{k}}
\left(  \rho_{k}-\alpha\right)  M_{k}^{\text{opt}}=P_{\text{succ}}\left(
M_{k}^{\text{opt}}\right)  -\alpha\dim\mathcal{H}\text{,}%
\end{align}
where $M_{k}^{\text{opt}}$ is some optimal POVM.

In order to control dimensional factors (and to consider ensembles on
infinite-dimensional Hilbert spaces) it is useful to introduce the following concept:

\begin{definition}
Let $\mathcal{E}=\left\{  \rho_{k}\right\}  $ be the ensemble of definition
\ref{main BC definition}, and let $p\in\left[  0,1\right]  $. The $\mathbf{p}%
$\textbf{-dimension }$\dim_{p}\left(  \mathcal{E}\right)  $ is the minimum
dimension of a subspace $\Lambda$ for which
\begin{equation}%
{\displaystyle\sum_{k}}
\left\Vert \left(  1-\Pi_{\Lambda}\right)  \rho_{k}\right\Vert _{1}\leq
p\text{,\label{condition wish could have in terms of sum of rhok}}%
\end{equation}
where $\Pi_{\Lambda}$ is the orthogonal projection onto $\Lambda$.
\end{definition}

\noindent\textbf{Remark:} Note that the inequality $\left(
\ref{condition wish could have in terms of sum of rhok}\right)  $ implies that%
\[
\operatorname*{Tr}\left(  1-\Pi_{\Lambda}\right)
{\displaystyle\sum}
\rho_{k}\leq\left\Vert
{\displaystyle\sum}
\left(  1-\Pi_{\Lambda}\right)  \rho_{k}\right\Vert _{1}\leq%
{\displaystyle\sum}
\left\Vert \left(  1-\Pi_{\Lambda}\right)  \rho_{k}\right\Vert _{1}\leq
p\text{.}%
\]

\begin{lemma}
\label{Lemma pdim well defined finite}For fixed $\mathcal{E}$, the function
$p\mapsto\dim_{p}\left(  \mathcal{E}\right)  $ is finite on $\left(
0,1\right]  $ and monotonically-decreasing on $\left[  0,1\right]  $.
\end{lemma}

\begin{proof}
The monotonicity of $p\mapsto\dim_{p}\left(  \mathcal{E}\right)  $ is
immediate from the definition. To prove finiteness for $p>0$, take spectral
decompositions $\rho_{k}=%
{\displaystyle\sum}
\lambda_{k\ell}\left\vert \psi_{k\ell}\right\rangle \left\langle \psi_{k\ell
}\right\vert $. For a finite subset $S$ of the $\left(  k,\ell\right)  $, let
$\Pi_{S}$ be the projection onto the linear span of the $\psi_{k\ell}$ with
$\left(  k,\ell\right)  \in S$. By the triangle inequality%
\[%
{\displaystyle\sum_{k}}
\left\Vert \left(  1-\Pi_{S}\right)  \rho_{k}\right\Vert _{1}\leq%
{\displaystyle\sum_{k\ell}}
\left\Vert \left(  1-\Pi_{S}\right)  \lambda_{k\ell}\left\vert \psi_{k\ell
}\right\rangle \left\langle \psi_{k\ell}\right\vert \right\Vert _{1}\leq%
{\displaystyle\sum_{\left(  k,\ell\right)  \notin S}}
\lambda_{k\ell}%
\]
Since $%
{\displaystyle\sum_{\left(  k,\ell\right)  \in S}}
\lambda_{k\ell}=1$, we may take a finite subset $S$ of the $\left(
k,\ell\right)  $ such that the right-hand side may be made smaller than $p$.
\end{proof}

We may now state a robust version of Theorem
\ref{theorem simplified optimality condition}:

\begin{theorem}
\label{theorem robust version of simplified optimality}Let $\left\{
M_{k}\right\}  $ be a POVM\ for distinguishing $\mathcal{E}$, let $L=%
{\displaystyle\sum}
M_{k}\rho_{k}$, and let $\{M_{k}^{\text{opt}}\}$ be an optimal measurement. Then

\begin{enumerate}
\item Assume that $\alpha>0$ is a scalar such that
\begin{equation}
\operatorname{Re}\left(  L\right)  \geq\rho_{k}-\alpha
\label{almost II star with aa error}%
\end{equation}
for all $k.$ Then for $p\in\left[  0,1/4\right)  $%
\begin{equation}
P_{\text{succ}}\left(  M_{k}\right)  \geq P_{\text{succ}}\left(
M_{k}^{\text{opt}}\right)  -\alpha\dim_{p}\left(  \mathcal{E}\right)
-4p\text{.}\label{inequality approx optimality}%
\end{equation}

\item Suppose that $\operatorname{Re}\left(  L\right)  \ngeq\rho_{\ell}$ for
some $\ell$. Then%
\begin{equation}
P_{\text{succ}}\left(  M_{k}\right)  \leq P_{\text{succ}}\left(
M_{k}^{\text{opt}}\right)  -\left(  \operatorname*{Tr}\left(  \left[
\rho_{\ell}-\operatorname{Re}\left(  L\right)  \right]  _{+}\right)  \right)
^{2},\label{upper bound on Psucc using
BC ideas}%
\end{equation}
where $\left[  \bullet\right]  _{+}$ is the positive part, defined in
definition \ref{def of positive part}.
\end{enumerate}
\end{theorem}

\subsection{Discussion of Theorem
\ref{theorem robust version of simplified optimality}}

The small-$\alpha$ case of Part 1 addresses the case where $\left\{
M_{k}\right\}  $ nearly-satisfies condition II. In particular, $\left(
\ref{inequality approx optimality}\right)  $ implies that $P_{\text{succ}%
}\left(  M_{k}\right)  \geq P_{\text{succ}}\left(  M_{k}^{\text{opt}}\right)
-\varepsilon$ if
\begin{equation}
\alpha<\sup_{p\in\left[  0,\varepsilon/4\right]  }\frac{\varepsilon-4p}%
{\dim_{p}\left(  \mathcal{E}\right)  }%
\text{.\label{condition to make Pfail nearly optimal}}%
\end{equation}
The following example shows that the dependence of this expression on
$\mathcal{E}$ may not be removed except (in the finite-dimensional case) by
introducing dimensional factors:

\begin{example}
Let $m$ be a positive integer, and let $\mathcal{E}$ be the $m$-state ensemble
on $\mathbb{C}^{m}$ defined by $\rho_{k}=\left\vert k\right\rangle
\left\langle k\right\vert /m.$ Set $M_{k}=\left\vert k+1\right\rangle
\left\langle k+1\right\vert $, using addition mod $m$. Then one has
$P_{\text{succ}}\left(  M_{k}\right)  =0$ and $P_{\text{succ}}\left(
M_{k}^{\text{opt}}\right)  =1$, but inequality $\left(
\ref{almost II star with aa error}\right)  $ holds for $\alpha=1/m,$ which
approaches $0$ as $m\rightarrow\infty$.
\end{example}

\subsection{Proof of part 1 of Theorem
\ref{theorem robust version of simplified optimality}}

\begin{proof}
Let $\Pi$ be an orthogonal projection, and set $\Pi^{\bot}=%
\openone
-\Pi$. Then%
\begin{equation}
P_{\text{succ}}\left(  M_{k}\right)  =\operatorname*{Tr}\left(  \Pi
\operatorname{Re}\left(  L\right)  \Pi\right)  +\operatorname*{Tr}\left(
\Pi^{\bot}\operatorname{Re}\left(  L\right)  \Pi^{\bot}\right)
\text{.\label{eq break with projection}}%
\end{equation}
Using equations $\left(  \ref{inequality trace less than i1 norm}\right)
$-$\left(  \ref{fact conjugate a positve operator}\right)  $ to estimate the
first term,%
\begin{align}
\operatorname*{Tr}\Pi\operatorname{Re}\left(  L\right)  \Pi &
=\operatorname*{Tr}%
{\displaystyle\sum_{k}}
\operatorname{Re}\left(  L\right)  \times\Pi M_{k}^{\text{opt}}\Pi\nonumber\\
& \geq\operatorname*{Tr}%
{\displaystyle\sum}
\left(  \rho_{k}-\alpha\right)  \times\Pi M_{k}^{\text{opt}}\Pi\nonumber\\
& =P_{\text{succ}}\left(  M_{k}^{\text{opt}}\right)  -\alpha\operatorname*{Tr}%
\left(  \Pi\right)  +\operatorname*{Tr}\left(
{\displaystyle\sum}
\left(  \Pi\rho_{k}\Pi-\rho_{k}\right)  M_{k}^{\text{opt}}\right) \nonumber\\
& \geq P_{\text{succ}}\left(  M_{k}^{\text{opt}}\right)  -\alpha
\operatorname*{Tr}\left(  \Pi\right)  -%
{\displaystyle\sum}
\left\Vert \rho_{k}-\Pi\rho_{k}\Pi\right\Vert _{1}\text{.}%
\label{first term estimate in dim free conv estimate}%
\end{align}
But%
\begin{align}%
{\displaystyle\sum_{k}}
\left\Vert \rho_{k}-\Pi\rho_{k}\Pi\right\Vert _{1}  & =%
{\displaystyle\sum_{k}}
\left\Vert \Pi^{\bot}\rho_{k}+\rho_{k}\Pi^{\bot}+\Pi^{\bot}\rho_{k}\Pi^{\bot
}\right\Vert _{1}\nonumber\\
& \leq3%
{\displaystyle\sum_{k}}
\left\Vert \Pi^{\bot}\rho_{k}\right\Vert _{1}%
\text{.\label{I1 estimate in dim free conv lemma}}%
\end{align}
Using $\left(  \ref{inequality l1 linfinity holder}\right)  $ to estimate the
second term of $\left(  \ref{eq break with projection}\right)  $,%
\begin{align}
\left\vert \operatorname*{Tr}\left(  \Pi^{\bot}\operatorname{Re}\left(
L\right)  \Pi^{\bot}\right)  \right\vert  & \leq\frac{1}{2}\left\Vert
{\displaystyle\sum}
\Pi^{\bot}\rho_{k}M_{k}\Pi^{\bot}+\Pi^{\bot}M_{k}\rho_{k}\Pi^{\bot}\right\Vert
_{1}\nonumber\\
& \leq%
{\displaystyle\sum}
\left\Vert \Pi^{\bot}\rho_{k}\right\Vert _{1}\label{second
term estimate in dim free conv est}%
\end{align}
Putting $\left(  \ref{eq break with projection}\right)  -\left(
\ref{second term estimate in dim free conv est}\right)  $ together gives%
\begin{equation}
P_{\text{succ}}\left(  M_{k}\right)  \geq P_{\text{succ}}\left(
M_{k}^{\text{opt}}\right)  -\alpha\operatorname*{Tr}\left(  \Pi\right)  -4%
{\displaystyle\sum}
\left\Vert \Pi^{\bot}\rho_{k}\right\Vert _{1}%
\text{.\label{bound on Psucc from below before putting in dimension}}%
\end{equation}
The bound $\left(  \ref{inequality approx optimality}\right)  $ follows by
picking $\Pi$ to minimize $\operatorname*{Tr}\left(  \Pi\right)  $ when the
last term of $\left(
\ref{bound on Psucc from below before putting in dimension}\right)  $ is
constrained to be less than $p$. (By Lemma
\ref{Lemma pdim well defined finite} such $\Pi$ of finite rank always exist.)
\end{proof}

\subsection{Proof of part 2 of Theorem
\ref{theorem robust version of simplified optimality}}

\begin{definition}
Let $\left\{  M_{k}\right\}  $ be a POVM for distinguishing the ensemble
$\mathcal{E}$ of definition \ref{main BC definition}, let $X\leq2\times%
\openone
$ be a positive semidefinite operator on $\mathcal{H}$, and let $\ell\in K$.
Then the \textbf{Barnett -Croke} \textbf{modification }of $\left\{
M_{k}\right\}  $ is defined by%
\[
M_{k}\left(  X,\ell\right)  =\left(  1-X\right)  M_{k}\left(  1-X\right)
+\delta_{k\ell}\left(  2X-X^{2}\right)  \text{.}%
\]

\end{definition}

\noindent\textbf{Remark:} Note that since $0\leq2X-X^{2}$ for $0\leq
X\leq2\times%
\openone
$, for each $\ell$ the set $\left\{  M_{k}\left(  X,\ell\right)  \right\}  $
forms a POVM. Barnett and Croke \cite{Barnett Croke on the conditions for
discrimination between quantum states with minimum error} considered the case
$X=\varepsilon\left\vert \psi\right\rangle \left\langle \psi\right\vert $,
where $\psi$ is a unit vector satisfying the eigenvalue equation
\[
\left(  \rho_{\ell}-\operatorname{Re}\left(  L\right)  \right)  \left\vert
\psi\right\rangle =-\lambda\left\vert \psi\right\rangle \text{,}%
\]
with $\lambda>0$. They showed that%
\[
\left.  \frac{d}{d\varepsilon}\right\vert _{\varepsilon=0}P_{\text{succ}%
}\left(  M_{k}\left(  X,\ell\right)  \right)  =2\lambda>0\text{.}%
\]
In order to complete the proof of part 2 of Theorem
\ref{theorem robust version of simplified optimality}, it suffices to turn
this perturbative argument into an estimate.

\bigskip

\begin{proof}
[Proof of part 2 of Theorem
\ref{theorem robust version of simplified optimality}]Let $\Pi_{+}$ be the
positive projection $\left(  \ref{def of positive projection}\right)  $%
\begin{equation}
\Pi_{+}=\chi_{+}\left(  \rho_{\ell}-\operatorname{Re}\left(  L\right)
\right)  \text{.\label{define pi plus}}%
\end{equation}
Then for $\alpha\in\left[  0,2\right]  $,%
\begin{align}
P_{\text{succ}}\left(  M_{k}\left(  \alpha\Pi_{+},\ell\right)  \right)    &
=P_{\text{succ}}\left(  M_{k}\right)  +2\alpha\operatorname*{Tr}\left[
\left(  \rho_{\ell}-\operatorname{Re}\left(  L\right)  \right)  \times\Pi
_{+}\right]  \nonumber\\
& -\alpha^{2}\operatorname*{Tr}\left(  \Pi_{+}\rho_{\ell}\right)  +\alpha
^{2}\operatorname*{Tr}%
{\displaystyle\sum}
\Pi_{+}M_{k}\Pi_{+}\rho_{k}\\
& \geq P_{\text{succ}}\left(  M_{k}\right)  +2\alpha\operatorname*{Tr}\left(
\left[  \rho_{\ell}-\operatorname{Re}\left(  L\right)  \right]  _{+}\right)
-\alpha^{2},\label{nice bound generalized BC}%
\end{align}
where we have used cyclicity of the trace and$\left(
\ref{fact yuen trace of P1P2}\right)  -\left(
\ref{fact conjugate a positve operator}\right)  $.

Note that if $\operatorname*{Tr}\left(  \left[  \rho_{\ell}-\operatorname{Re}%
\left(  L\right)  \right]  _{+}\right)  >1$ then%
\[
P_{\text{succ}}\left(  M_{k}\left(  \Pi_{+},\ell\right)  \right)
=P_{\text{succ}}\left(  M_{k}\right)  +2\operatorname*{Tr}\left(  \left[
\rho_{\ell}-\operatorname{Re}\left(  L\right)  \right]  _{+}\right)  -1>1,
\]
giving a contradiction. In particular, we may set
\begin{equation}
\alpha=\operatorname*{Tr}\left(  \left[  \rho_{\ell}-\operatorname{Re}\left(
L\right)  \right]  _{+}\right)  \in\left[  0,1\right]  ,\label{define ee}%
\end{equation}
maximizing the RHS of $\left(  \ref{nice bound generalized BC}\right)  $ over
$\alpha\in\left[  0,1\right]  $. This gives%
\begin{equation}
P_{\text{succ}}\left(  M_{k}\left(  \alpha\Pi_{+},\ell\right)  \right)  \geq
P_{\text{succ}}\left(  M_{k}\right)  +\left(  \operatorname*{Tr}\left(
\left[  \rho_{\ell}-\operatorname{Re}\left(  L\right)  \right]  _{+}\right)
\right)  ^{2}.\label{bound on iterate that has square in it}%
\end{equation}

\end{proof}

\section{Barnett-Croke iteration\label{section algorithm}}

In this section we show how to convert Barnett and Croke's perturbative proof
into an algorithm for computing optimal measurements. Although the success
rate of poorly-chosen iterations might fail to actually converge to that of an
optimal measurement,\footnote{In is asserted in \cite{Jezek Rehacek and
Fiurasek Finding optimal strategies for minimum error quantum state
discrimination} that the algorithm of \cite{Helstrom Bayes cost reduction}
suffers this fate.} the following sequence does not exhibit this malady:

\begin{definition}
\label{definition iterative scheme}Let $\left\{  M_{k}\right\}  $ be a POVM
for distinguishing the ensemble $\mathcal{E}$ of definition
\ref{main BC definition}, and chose $\ell$ to maximize%
\begin{equation}
\alpha=\operatorname*{Tr}\left[  \rho_{\ell}-\operatorname{Re}\left(
L\right)  \right]  _{+}\text{.\label{first maximum}}%
\end{equation}
Then the \textbf{iterate }of $\left\{  M_{k}\right\}  $ is the POVM%
\begin{equation}
M_{k}^{+}=M_{k}\left(  \alpha\chi_{+}\left(  \rho_{\ell}-\operatorname{Re}%
\left(  L\right)  \right)  ,\ell\right)
\text{,\label{equation defining Mkplus}}%
\end{equation}
where $\left[  \bullet\right]  _{+}$ and $\chi_{+}$ are defined in $\left(
\ref{def of positive pt}\right)  -\left(  \ref{def of positive projection}%
\right)  $. For a given measurement $\{M_{k}^{\left(  0\right)  }\},$
recursively define the \textbf{iterative series }$\{M_{k}^{\left(  n\right)
}\}_{n\geq1}$ by\footnote{Faster convergence can be obtained by replacing
$\alpha$ by $\beta$ in equation $\left(  \ref{equation defining Mkplus}%
\right)  ,$ where $\beta\in\left[  0,2\right]  $ is chosen to maximize
$P_{\text{succ}}\left(  M_{k}^{+}\right)  $, which is quadratic in $\beta$.}
\begin{equation}
M_{k}^{\left(  n+1\right)  }=\left(  M_{k}^{\left(  n\right)  }\right)
^{+}\text{.\label{eq defining iterative sequence}}%
\end{equation}

\end{definition}

\noindent\textbf{Remark:} An index $\ell$ maximizing $\left(
\ref{first maximum}\right)  $ exists using minimax principle (Theorem XIII.1
of \cite{Reed and Simon IV}) and the fact that $\operatorname*{Tr}%
{\displaystyle\sum}
\rho_{\ell}=1$.

The proof of part II of Theorem
\ref{theorem robust version of simplified optimality} actually proved the
following stronger result:

\begin{theorem}
\label{theorem amount iterate increases}The above iteration monotonically
increases success rate. In particular, for an arbitrary POVM $\left\{
M_{k}\right\}  $ the set $\left\{  M_{k}^{+}\right\}  $ is a well-defined
POVM, and%
\begin{equation}
P_{\text{succ}}\left(  M_{k}^{+}\right)  \geq P_{\text{succ}}\left(
M_{k}\right)  +\max_{\ell}\left(  \operatorname*{Tr}\left(  \left[  \rho
_{\ell}-\operatorname{Re}\left(  L\right)  \right]  _{+}\right)  \right)
^{2}\text{.}\label{inequality iteration pretty good}%
\end{equation}

\end{theorem}

We now show that the iterative scheme of definition
$\ref{definition iterative scheme}$ approaches optimality:

\begin{theorem}
Let $M_{k}^{\left(  0\right)  }$ be an arbitrary starting POVM for the
iterative series $\left(  \ref{eq defining iterative sequence}\right)  $.
Then
\begin{equation}
\lim_{n\rightarrow\infty}P_{\text{succ}}(M_{k}^{\left(  n\right)
})=P_{\text{succ}}(M_{k}^{\text{opt}})\text{,}%
\end{equation}
where $M_{k}^{\text{opt}}$ is an optimal measurement.
\end{theorem}

\begin{proof}
Let $\varepsilon>0$ be arbitrary. We seek an $N>0$ such that%
\begin{equation}
n>N\Rightarrow P_{\text{succ}}\left(  M_{k}^{\left(  n\right)  }\right)  \geq
P_{\text{succ}}\left(  M_{k}^{\text{opt}}\right)  -\varepsilon\text{.}%
\label{claim BC iteration converges}%
\end{equation}
Set%
\[
L^{\left(  n\right)  }=%
{\displaystyle\sum}
M_{k}^{\left(  n\right)  }\rho_{k}.
\]
By equation $\left(  \ref{condition to make Pfail nearly optimal}\right)  $
and the monotonicity of $n\mapsto P_{\text{succ}}\left(  M_{k}^{\left(
n\right)  }\right)  $, it suffices to find a $n\leq N$ such that%
\begin{equation}
\operatorname{Re}\left(  L^{\left(  n\right)  }\right)  \geq\rho_{\ell}%
-\Delta\label{in BC iter prove this}%
\end{equation}
for all $\ell$, where $\Delta$ is any real number satisfying\footnote{In
finite dimensions, one may take $\Delta=\varepsilon/\dim\mathcal{H}%
\leq\varepsilon/\dim\left(  \operatorname{supp}\left(  \mathcal{E}\right)
\right)  $, corresponding to $p=0$.}
\begin{equation}
0<\Delta\leq\sup_{p\in\left[  0,\varepsilon/4\right]  }\frac{\varepsilon
-4p}{\dim_{p}\left(  \mathcal{E}\right)  }.
\end{equation}
We claim that $N=\Delta^{-2}$ suffices.

Assume that%
\[
\max_{\ell}\operatorname*{Tr}\left(  \left[  \rho_{\ell}-\operatorname{Re}%
\left(  L^{\left(  n\right)  }\right)  \right]  _{+}\right)  >\Delta
\]
for all $n\leq N$. By Theorem $\ref{theorem amount iterate increases},$%
\[
P_{\text{succ}}\left(  M_{k}^{\left\lceil N\right\rceil +1}\right)
>N\times\Delta^{2}\geq1,
\]
yielding a contraction.

It follows that%
\[
\max_{\ell}\operatorname*{Tr}\left(  \left[  \rho_{\ell}-\operatorname{Re}%
\left(  L^{\left(  n\right)  }\right)  \right]  _{+}\right)  \leq\Delta
\]
for some $n\leq N.$ The inequality $\left(  \ref{in BC iter prove this}%
\right)  $ follows from the observation that%
\[
A\leq\operatorname*{Tr}\left(  \left[  A\right]  _{+}\right)  \times%
\openone
\text{,}%
\]
for $A=\rho_{\ell}-\operatorname{Re}\left(  L^{\left(  n\right)  }\right)  .$
\end{proof}

\section{Conclusion}

Using non-optimality estimates in terms of quantitative violation of Holevo's
simplified optimal measurement condition, we have converted Barnett and
Croke's perturbative proof into a conceptually-simple iterative scheme for
computing optimal measurements. This iteration approaches the optimal success
rate even in the case of infinite-dimensions and infinite ensemble
cardinality. It would be interesting to try to improve the non-optimality
bounds of Theorem \ref{theorem robust version of simplified optimality}, and
to study the convergence rate of this iteration in more detail.

\bigskip

\noindent\textbf{Acknowledgements:} The author would like to thank A. Holevo
for pointing out a useful reference, and Arthur Jaffe for his encouragement.

\end{document}